\documentclass[aps,showpacs,pra,twocolumn,superscriptaddress]{revtex4-1}

\usepackage{exscale}
\usepackage{bbm}
\usepackage{graphicx}
\usepackage{amsmath}
\usepackage{latexsym}
\usepackage{amsfonts}
\usepackage{amssymb}
\usepackage{times}
\usepackage[T1]{fontenc}
\usepackage{amsthm}
\usepackage{enumerate}
\usepackage{bbold}
\usepackage{color}
\usepackage[colorlinks=true,citecolor=blue,urlcolor=blue]{hyperref}
\usepackage{mathtools}
\usepackage{changes}
\usepackage[normalem]{ulem}

\definecolor{myurlcolor}{rgb}{0,0,0.4}
\definecolor{mycitecolor}{rgb}{0,0.5,0}
\definecolor{myrefcolor}{rgb}{0.5,0,0}
\usepackage{hyperref}
\hypersetup{colorlinks,
linkcolor=myrefcolor,
citecolor=mycitecolor,
urlcolor=myurlcolor}

\usepackage[draft]{fixme}
\usepackage{tikz}
\usepackage{subfigure}
\usepackage{amsmath,bbm}
\usepackage{graphicx}
\usepackage{amsfonts}
\usepackage{amssymb}
\usepackage{amsmath, amssymb, amsthm,verbatim,graphicx,bbm}
\usepackage{mathrsfs}
\usepackage{color,xcolor,longtable}

\usepackage{mathpazo,palatino} 

\usetikzlibrary{patterns}

\newcommand{\one}{\leavevmode\hbox{\small1\normalsize\kern-.33em1}}
\def\tr{\mbox{tr}}

\def\beq{\begin{equation}}
\def\eeq{\end{equation}}
\def\be{\begin{equation}}
\def\ee{\end{equation}}
\def\ben{\begin{eqnarray}}
\def\een{\end{eqnarray}}
\def\eea{\end{array}}
\def\bea{\begin{array}}

\newcommand{\bei}{\begin{itemize}}
\newcommand{\eei}{\end{itemize}}
\newcommand{\ket}[1]{|#1\rangle}
\newcommand{\bra}[1]{\langle#1|}

\newcommand{\proj}[1]{\ket{#1}\!\bra{#1}}

\def\pcal{{\cal P}}

\renewcommand{\emph}[1]{\textbf{#1}}

\theoremstyle{plain}
\newtheorem{fakt}{Fact}
\newtheorem*{thm*}{Theorem}

\theoremstyle{definition}

\theoremstyle{remark}

\usepackage{float}

\begin{document}

\title{Simple sufficient condition for subspace to be completely or genuinely entangled}

\author{Maciej Demianowicz}
    \affiliation{{\it\small Institute of Physics and Applied Computer Science, Faculty of Applied Physics and Mathematics,
Gda\'nsk University of Technology, Narutowicza 11/12, 80-233 Gda\'nsk, Poland}}
\email{maciej.demianowicz@pg.edu.pl}
\author{Grzegorz Rajchel-Mieldzio{\'c}}
\affiliation{Center for Theoretical Physics, Polish Academy of Sciences, Aleja Lotnik\'{o}w 32/46, 02-668 Warsaw, Poland}
\author{Remigiusz Augusiak}
\affiliation{Center for Theoretical Physics, Polish Academy of Sciences, Aleja Lotnik\'{o}w 32/46, 02-668 Warsaw, Poland}

\date{\today}

\begin{abstract}
We introduce a simple sufficient criterion, which allows one to tell whether a subspace of a bipartite or multipartite Hilbert space is entangled. 
The main ingredient of our criterion is a bound on the minimal entanglement
of a subspace in terms of entanglement of vectors spanning that subspace expressed for geometrical measures of entanglement. The criterion is applicable to both completely and genuinely entangled subspaces.
We explore its usefulness in several important scenarios. Further, an entanglement criterion for mixed states following directly from the condition is stated. As an auxiliary result we provide a formula for the generalized geometric measure of entanglement of the $d$--level Dicke states.
\end{abstract}

\maketitle

\section{Introduction}
Quantum entanglement is one of the  central notions of modern physics and technology and it has been a subject of intensive efforts  in the recent decades towards its complete characterization.  
An important line of research in the field  is the one aiming at describing properties of completely \cite{Parthasarathy,Bhat,WalgateScott}  or genuinely entangled subspaces \cite{upb-to-ges,approach,Agrawal,4x4,antipin}, which are those subspaces of the composite Hilbert spaces that contain only entangled
or genuinely multiparty entangled (GME) states, respectively. This was primarily motivated by their theoretical importance as any state with support in an entangled subspace is necessarily entangled but they are also relevant from the practical point of view, e.g., in quantum error correction \cite{HuberGrassl} or in protocols where the existence of entanglement needs to be certified, for example super-dense coding \cite{Horodecki_2012}.

A notorious problem 
in this domain is to determine whether a given subspace is completely or genuinely entangled.
Certification of subspace entanglement is a hard task as it requires proving that any superposition of states in a subspace is entangled, or, phrasing it differently, the minimal subspace entanglement is non--vanishing. One way to achieve this is to consider lower bounds on the entanglement of superposed states from a subspace and check under which circumstances they give a positive number. The research on such bounds  has its own history \cite{Ma-negativity-1,Ma-negativity-2,Ma-GME-superposition,Akhta,Cavalcanti,Xiong,NisetCerf,OuFan,Song}, which was initiated in Ref. \cite{LPS}. Following this path, the authors of \cite{Gour} proved that the minimal subspace entanglement, as measured by the Schmidt rank, in the bipartite scenario  can be lower--bounded solely in terms of the entanglement of the basis states. Another approach to bounding the minimal subspace entanglement was put forward in Ref. \cite{ent-ges} in terms of semi--definite programs (SDPs). There, also a see-saw algorithm to compute it exactly was proposed. Both methods, however, do not exploit the entanglement of the basis states, but rather of the projection onto the subspace and are more involved. It is finally interesting to notice that it is possible to certify genuine entanglement of a subspace in a device-independent way, and the first strategies to this end have recently been presented in \cite{DIsubspaces,DIsubspaces2}.

In this work, building upon the lower bound on the geometric measure of entanglement (GM) of a superposition from \cite{Song},  we develop a very simple, yet nontrivial, sufficient condition for a subspace to 
exhibit certain entanglement properties such as being completely or genuinely entangled. The criterion involves only entanglement of the orthogonal spanning vectors
and is stated for a wide class of entanglement quantifiers, called geometrical measures (of which the GM is only an example). 
Its value lies in broad applicability as it works both in the bipartite and multipartite setups for any local dimensions (except for systems with a qubit subsystem in the former case), and is not limited by the Schmidt ranks of the basis states.
We provide several illustrative examples to investigate the power of our criterion. In particular, we discuss subspaces spanned by the (generalized) 
Greenberger–Horne–Zeilinger (GHZ), Dicke, and absolutely maximally entangled (AME) states, and the antisymmetric subspace. Furthermore, we show how this condition leads directly to an entanglement criterion for mixed states, determined by the entanglement of the states in the mixture.

\section{Preliminaries}\label{preliminaria}

Before presenting our main result we introduce relevant
notions and terminology. 


\subsection{Entanglement and separability} 

Let us begin with the simplest bipartite case and consider a Hilbert space $\mathcal{H}_{2,d}=\mathbbm{C}^d\otimes\mathbbm{C}^d$. 
A state $\ket{\psi}\in\mathcal{H}_{2,d}$ is termed \textit{separable} or \textit{product} iff $\ket{\psi}=\ket{\psi_1}\otimes\ket{\psi_2}$ for some
pure states $\ket{\psi_i}\in\mathbbm{C}^d$; otherwise it is called \textit{entangled}.
A particularly useful tool for the characterization of entanglement of bipartite pure states is {\it the Schmidt decomposition}: any $\ket{\psi}\in\mathcal{H}_{2,d}$ can be written as
\begin{equation}\label{schmidt}
    \ket{\psi}=\sum_{i=1}^{r}\lambda_i\ket{e_i}\ket{f_i},
\end{equation}
where the coefficients $\lambda_i$ are positive numbers that can be ordered as $\lambda_1\geq\lambda_2\geq\ldots\geq\lambda_r$, whose squares sum up to one, $\{\ket{e_i}\}_i$ and $\{\ket{f_i}\}_i$ are orthonormal bases of the local Hilbert spaces and $r\leq d$ is called {\it the Schmidt rank} of $\ket{\psi}$. The Schmidt decomposition allows to  decide efficiently whether a pure state is entangled or product and to determine the actual dimension of the local Hilbert spaces needed to support that state. Let us denote  by $Sch_r$ the set of all pure states
from $\mathcal{H}_{2,d}$ of the Schmidt rank at most $r$.  Clearly, $Sch_1 \subset Sch_2\subset  \dots \subset Sch_d$.

The multipartite scenario features a whole variety of different forms of entanglement. A particularly useful notion to tame them is that of the $k$-{\it producibility} \cite{GuhneBriegel}. An $N$--partite pure state $\ket{\psi}\in\mathcal{H}_{N,d}=(\mathbbm{C}^d)^{\otimes N}$ is termed {\it $k$-producible} if it can be written as 
\begin{equation}
\ket{\psi}=\ket{\phi_1}\otimes \ldots \otimes \ket{\phi_M},
\end{equation}
with each $\ket{\phi_i}$ corresponding to at most $k$ particles and $M\leq N$. Let us denote by $\mathcal{P}_k$ the set of all states that are $k$-producible. Then, $\mathcal{P}_N$ is the set of all pure states. Clearly, the following inclusions  hold true
$\mathcal{P}_1\subset \mathcal{P}_2\subset\ldots\subset \mathcal{P}_N$. Vectors from $\mathcal{P}_1$ are said to be {\it fully product}, whereas states belonging to $\pcal_k \setminus \pcal_1$ for  $k\ne 1$ are {\it entangled}, i.e., not fully product (although they still might be product across some cuts).  In particular, those belonging to $\mathcal{P}_{N}\setminus \mathcal{P}_{N-1}$ are called {\it genuinely multiparty entangled} (GME) as they do not display any form of separability. Another way of saying that a state is GME is that it is not product across any bipartition of the parties, or, equivalently, it is not {\it biproduct}.
Importantly, if $\ket{\psi}$ belongs to $\mathcal{P}_k$ but not to $\mathcal{P}_{k-1}$ we know that 
at least $k$ particles share genuinely multipartite entanglement; one also says that the 
{\it entanglement depth} of such a state is at least $k$.

Undoubtedly, the most widespread example of a GME state is the Greenberger-Zeilinger-Horne (GHZ) state
%
%
\begin{equation}
\ket{\mathrm{GHZ}_{N,d}}=\frac{1}{\sqrt{d}}\sum_{i=0}^{d-1}\ket{i}^{\otimes N}.
\end{equation}
%
%
%
Another important class of GME states that are
often considered in the quantum information context 
are the symmetric $N$-qubit Dicke states 
\begin{equation}\displaystyle
\ket{D_{N,k}}=\frac{1}{\sqrt{\binom{N}{k}}}\sum_p\sigma_p\left(\ket{0}^{\otimes (N-k)}\ket{1}^{\otimes k}\right),
\end{equation}
where $\{\sigma_p\}$ is the set of all distinct permutations. In other words, the 
Dicke states are symmetrized versions of pure states in which $N-k$ particles are in 
the ground state $\ket{0}$, while the remaining $k$ particles are in the excited state $\ket{1}$. We also consider their generalizations to arbitrary local dimensions in further parts of the paper.

In the case of genuine multipartite entanglement it also makes sense to define further genuine multiparty entanglement of bounded Schmidt rank. Namely, a state $\ket{\psi}\in\mathcal{H}_{N,d}$ is said to 
have genuine multiparty entanglement of $r$--bounded Schmidt rank if it is GME and has the Schmidt rank at least some $r$  with respect to any bipartition. The corresponding set is denoted $\mathcal{G}_r$. Even further, in some cases it might be interesting to consider states which have the same Schmidt rank across any bipartition.

\subsection{Completely and genuinely entangled subspaces}
Recently, we have witnessed growing interest in the research on subspaces, not only single states, possessing certain entanglement properties. This stems from both their theoretical and practical importance. In this respect two classes of such subspaces stand out: completely and genuinely entangled ones, however, it is also meaningful to consider  subspaces with other properties considered above.

Let us formally introduce respective subspaces.
Consider a proper subspace $V\subset \mathcal{H}_{N,d}$. We call $V$ a \textit{completely entangled subspace (CES)} iff any vector belonging to it is entangled, or, in other words, $V$ contains no fully product vectors \cite{Parthasarathy,Bhat}. 
Well-known examples of such subspaces are those obtained from unextendible product bases (UPB) -- another interesting object introduced to construct entangled states with positive partial transpositions \cite{UPB}. A UPB is a set of fully product vectors, spanning a proper subspace of a given Hilbert space, with the property that there does not exist a fully product vector orthogonal to all the vectors from the set. It follows that every subspace complementary to a UPB is completely entangled.
 To give an example let us for a moment focus on $\mathcal{H}_{3,2}=\mathbbm{C}^{2}\otimes\mathbbm{C}^{2}\otimes \mathbbm{C}^{2}$ and consider the following product vectors from it:
%
$\ket{0,0,0}$,
$\ket{1,+,-}$,
$\ket{-,1,+}$,
$\ket{+,-,1}$,
%
where $\{\ket{0},\ket{1}\}$ and $\{\ket{+},\ket{-}\}$ with $\ket{\pm}=(\ket{0}\pm\ket{1})/\sqrt{2}$, are two orthonormal bases in $\mathbbm{C}^2$. It is direct to check that these vectors form a UPB in $\mathcal{H}_{3,2}$ and thus the four-dimensional subspace orthogonal to it is a CES.

Further, we call $V\subset \mathcal{H}_{N,d}$ a \textit{genuinely entangled subspace (GES)} iff any pure state from $V$ is GME \cite{Cubitt, upb-to-ges}. To give an example of a GES let us consider again the three-qubit Hilbert space $\mathcal{H}_{3,2}$ and the vectors
$\ket{\mathrm{GHZ_{3,2}}}$ and $\ket{D_{3,1}}$ (also called the
$W$ state). Due to the fact that they are both symmetric, the two-dimensional subspace spanned by them is symmetric too. Now, if a biproduct vector was to belong in this subspace, it would have to be symmetric, and consequently fully product, i.e., of the form $\ket{e}^{\otimes 3}$ with some arbitrary qubit vector $\ket{e}$. It is, however, not difficult to see that there do not exist
$a,b\in\mathbbm{C}$ $(|a|^2+|b|^2=1)$ such that $a\ket{\mathrm{GHZ}_{3,2}}+b\ket{W}=\ket{e}^{\otimes 3}$. Another example of a GES, which is well--known in the literature, is
the antisymmetric subspace of $\mathcal{H}_{N,d}$; this subspace is genuinely entangled because there do not exist product vectors that are antisymmetric. Importantly, however, one should bear in mind that these subspaces only exist if $d\geq N$ and they are of small dimensionality, while it is known how to construct large GESs efficiently for any $d$ and $N$ \cite{upb-to-ges} (see also \cite{approach, antipin}).

As we have mentioned at the beginning of this section, other notions of  entanglement defined for pure states can also  be meaningfully mapped to subspaces. We say that $V\subset \mathcal{H}_{N,d}$ is a subspace of entanglement depth at least $k$ iff it consists of pure states from $\pcal_N\setminus \pcal_{k-1}$; simply speaking, such a subspace contains only pure states whose entanglement depth is at least $k$.
We then define bipartite CESs with bounded or equal Schmidt rank to be those subspaces of $\mathcal{H}_{2,d}$, which contain only vectors having Schmidt rank, respectively, equal to or larger than $r$ \cite{Cubitt} or exactly $r$ \cite{westwick,Cubitt}.  Analogously, in the multipartite Hilbert spaces, we can define GESs of bounded or equal Schmidt rank, this time, however, requiring that the property holds with respect to any bipartition.

\subsection{Entanglement quantifiers}\label{miary}

To quantify entanglement we will use measures, which, due to their interpretation as distances, fall into a broad class of geometric measures. They are defined through the following general formula
\begin{equation}\label{miara-geometryczna}
\mathcal{E}(\ket{\psi})=1-\max_{\ket{\varphi}\in \mathcal{S}}|\langle\varphi|\psi\rangle|^2,
\end{equation}
where $\mathcal{S}$ is chosen accordingly to purposes $\mathcal{E}$ is supposed to serve.
Prominent representatives of the class are the geometric measure of entanglement (GM) \cite{GME1,GME2} and the generalized geometric measure of entanglement (GGM) \cite{ggm}, and they will be primarily  used measures in the discussion of applications of our result.

 First, we concentrate on the bipartite case and consider a pure state $\ket{\psi}\in\mathcal{H}_{2,d}$. The geometric measure of entanglement of $r$--bounded Schmidt rank 
  is defined by
\begin{equation}\label{r-schmidt}
E_r(\ket{\psi})=1-\max_{\ket{\varphi}\in Sch_{r-1}}|\langle\varphi|\psi\rangle|^2.
\end{equation}
with maximum over states of Schmidt rank at most $r-1$.
In the particular case of $r=2$, one recovers the  definition of the geometric measure of entanglement (GM), $E_{GM}$, in which 
the maximum is taken over all product, i.e., Schmidt rank one, vectors from $\mathcal{H}_{2,d}$ \cite{GME1,GME2}. 
For any $r$ it holds (see Appendix \ref{schmidt-proof})
\begin{equation}\label{KyFan}
    E_r(\ket{\psi})=1-(\lambda^2_1+\lambda^2_2+\ldots+\lambda^2_{r-1}),
\end{equation} 
where $\lambda_i$'s are the Schmidt coefficients of $\ket{\psi}$ [cf. (\ref{schmidt})].
The maximal value of $E_r$ over all states is $1-(r-1)/d$  and it is achieved by the maximally entangled state of two qudits $\ket{\Phi_d^+}=(1/\sqrt{d})\sum_{i=0}^{d-1} \ket{ii}$.

In the multipartite case, we  define the geometric measure of $k$-producibility 
\begin{equation}\label{k-producible}
    E^{producib}_{k}(\ket{\psi})=1-\max_{\ket{\varphi}\in \mathcal{P}_{k-1}}|\langle\varphi|\psi\rangle|^2,
\end{equation}
where the maximum is taken over all $(k-1)$-producible states. 
In particular, for $k=2$ the maximisation is performed over fully product pure states and we obatin the definition of the GM in the multipartite case (for $N=2$ this obviously reduces to the previously recalled definition of the GM). Let us write it out explicitly
\begin{equation}\label{}
E_{GM}(\ket{\psi})=1-\max_{\ket{\varphi_{prod.}}}|\bra{\varphi_{prod.}} \psi\rangle|^2.
\end{equation}
On the other hand, for $k=N$ in Eq. (\ref{k-producible}) we recover a definition
 of the generalized geometric measure of entanglement (GGM), $E_{GGM}$, which is more commonly stated in terms of maximisation over biproduct states:
 \begin{equation}\label{}
 E_{GGM}(\ket{\psi})=1-\max_{\ket{\varphi_{biprod.}}}|\bra{\varphi_{biprod.}} \psi\rangle|^2.
 \end{equation}
 The GGM
is designed to quantify genuine multiparty entanglement of a state and its vanishing implies that a given state is not GME.

Further, we
define the geometric measure of genuine multiparty entanglement of $r$--bounded Schmidt rank
as 
\begin{equation}
    E^{\mathrm{GME}}_{r}(\ket{\psi})=1-\max_{\ket{\varphi}\in \mathcal{G}_{r-1}}|\langle\varphi|\psi\rangle|^2
\end{equation}
with the maximisation over GME states with the Schmidt rank at most $r-1$.

Now, for any entanglement measure $\mathcal{E}$, one can define the corresponding notion of 
subspace entanglement. Precisely, for a given  subspace $V\subset \mathcal{H}_{N,d}$, 
we define the minimal subspace entanglement as (cf. \cite{Gour})
\begin{equation}
    \mathcal{E}_{\min}(V)=\min_{\ket{\psi}\in V}\mathcal{E}(\ket{\psi}).
\end{equation}
Thus, $\mathcal{E}_{\min}(V)$ is the entanglement of
the least entangled state in $V$ according to a given quantifier $\mathcal{E}$.
In particular, $\mathcal{E}$ can be chosen to be one of the introduced entanglement measures:
$E_r$, $E_{k}^{producib}$, $E_r^{\mathrm{GME}}$, or, more specifically, $E_{GM}$ or $E_{GGM}$. In the latter case, non--vanishing of the minimal subspace entanglement implies that a given subspace is a CES (the GM case) or a GES (the GGM case).

We must note here that the restriction to equal local dimensions was made here only for a simpler presentation and all the measures are defined in the same manner in the case of unequal 
dimensions.
\section{Main result: sufficient condition for a subspace to be entangled}
We can now move on to our results. Let us begin with
a lower bound on the  entanglement of  a superposition of pure states. This bound has been derived in  Ref. \cite{Song} for the GM but it is easy to see that it holds in general for any geometric measure (\ref{miara-geometryczna}). We have the following.
\begin{fakt}\label{fakt1}
Let $\ket{\phi_i}$, $i=1,\ldots,k$, be pairwise orthogonal pure states from $\mathcal{H}_{N,d}$ and let 
%
$\ket{\Psi}=\sum_{i=1}^k \alpha_i\ket{\phi_i}$
%
be their arbitrary superposition with $\alpha_i\in\mathbbm{C}$ such that 
$|\alpha_1|^2+\ldots+|\alpha_k|^2=1$. Then, for any $\mathcal{E}$ defined in Eq. (\ref{miara-geometryczna}), the following inequality holds true 
\begin{eqnarray}\label{bound-superpozycja}
\hspace{-0.5cm}\mathcal{E}(\ket{\Psi})&\geq& \sum_{i=1}^{k}|\alpha_i|^2 \mathcal{E}(\ket{\phi_i})\\
&&-2\sum_{i<j}|\alpha_i\alpha_j|
\sqrt{1-\mathcal{E}(\ket{\phi_i})}\sqrt{1-\mathcal{E}(\ket{\phi_j})}. \nonumber
\end{eqnarray}
\end{fakt}
\begin{proof}
To make the paper self-contained a proof of this inequality was added in Appendix \ref{AppA}.
\end{proof}

Using inequality (\ref{bound-superpozycja}) we  derive the main ingredient of our condition, which is a lower bound on the minimal subspace entanglement in terms of the entanglement of vectors spanning that subspace. Precisely, we have the following fact.

\begin{fakt}\label{fakt-glowny}
Consider a subspace $V$ spanned by $k$ pairwise orthogonal pure states $\ket{\phi_i}\in\mathcal{H}_{N,d}$. 
Then, the minimal subspace entanglement of $V$ is bounded from 
below in terms of $\mathcal{E}(\ket{\phi_i})$ as 
\begin{equation}\label{bound-min-subspace}
\mathcal{E}_{\min}(V)\geq \sum_{i=1}^{k}\mathcal{E}(\ket{\phi_i})-(k-1),
\end{equation}
where, as before, $\mathcal{E}$ can be any geometric quantifier of the form (\ref{miara-geometryczna}).
\end{fakt}
\begin{proof}
Two proofs of this fact are presented in Appendix~\ref{AppB}.
\end{proof}

Inequality (\ref{bound-min-subspace}) can be then used to 
formulate a simple sufficient condition for a subspace $V$ to have certain entanglement property such as being completely or genuinely entangled. 

\begin{fakt} \label{criterion}
(i) Given an $m$--dimensional bipartite subspace $V\subset\mathcal{H}_{2,d}$, if there is an orthonormal basis
$\{\ket{\phi_i}\}_{i=1}^m$ for $V$ such that 
\begin{equation}\label{kryterium}
\sum_{i=1}^m E_{r}(\ket{\phi_i})-(m-1)>0
\end{equation}
with $r \ge 2$, then $V$ is a CES containing only vectors with the Schmidt rank at least $r$.  

(ii) Given an $m$--dimensional multipartite subspace $V\subset\mathcal{H}_{N,d}$,  if there is an orthonormal basis
$\{\ket{\phi_i}\}_{i=1}^m$ for $V$ such that 
\begin{equation}\label{kryterium-multiparty}
\sum_{i=1}^m E_{k}^{producib}(\ket{\phi_i})-(m-1)>0,
\end{equation}
with $k\ge 2$, then $V$ is a CES. Morevoer, if (\ref{kryterium-multiparty}) holds for $k=N$, then $V$ is a GES. 
\end{fakt}

We must emphasize here that, although stated for equal local dimensions, the result also holds in the general case.

An important remark here is that Fact \ref{criterion} refers to any basis and in some cases it may be beneficial to check the conditions for different bases. 
In the next section we give an elementary example showing that indeed such change of basis may help detect entanglement of a subspace. On the other hand, if a measure of the basis states is not available one can use lower bounds on their individual entanglements [these are always easily obtained, e.g., by choosing particular states instead of the optimal ones in (\ref{miara-geometryczna})], which would clearly result in a weaker detectability of a given subspace.

A natural consequence of Fact \ref{criterion} is the following entanglemement criterion.

\begin{fakt}\label{ent-criterion}
	Given is a state $\varrho=\sum_{i=1}^k p_i \proj{\psi_i} $ acting on $\mathcal{H}_{N,d}$ with orthogonal $\psi_i$'s. 
	 If $\sum_{i=1}^k \mathcal{E}(\ket{\psi_i})-(k-1)>0$, where  $\mathcal{E}$ is a geometric measure of the form (\ref{miara-geometryczna}), which is non--vanishing only on entangled states (i.e., vanishes only on fully product states), then the state is entangled. If $\mathcal{E}$ in the above condition is a geometric measure of genuine multiparty entanglement, then $\varrho$ is genuinely multiparty entangled.
\end{fakt}
\begin{proof}
	Given the premise, by Fact \ref{fakt-glowny}, we conlude that the subspace $\mathcal{V}=\mathrm{span} \{\ket{\psi_i}\}$ is entangled. Further, any state supported on an entangled subspace is entangled. Obviously, if $\mathcal{V}$ is a GES the state is GME.
\end{proof}

Let us now move to applications of Fact \ref{criterion} and
provide  a few examples of entangled subspaces that are correctly identified as such with the aid of our bound and also discuss its applicability in general. We mainly focus on the most widespread measure -- the (generalized) geometric measure of entanglement.

\section{Applications in bipartite case}

Consider first subspaces of  $\mathcal{H}_{2 \times d}=\mathbbm{C}^2\otimes \mathbbm{C}^d$. We can limit ourselves here to $d>2$ since, as it is well known, there are no CESs in a system of two qubits -- any two-dimensional (or more) subspace of $\mathbbm{C}^2\otimes \mathbbm{C}^2$ contains a product vector.
Now, any entangled state in $\mathcal{H}_{2 \times d}$
is of Schmidt rank two (in other words, states from $\mathcal{H}_{2 \times d}$ are in fact two--qubit states) and the maximal value of the geometric measure of entanglement for a Schmidt rank-two state is $1/2$. This means that even if we take $k$ maximally entangled vectors from
$\mathcal{H}_{2 \times d}$  achieving this value,  we can make the right-hand side of the bound greater than zero
only if $k<2$, which is a trivial case of a single state CES. This limits utility of the criterion to systems with at least qutrit subsystems.

Let us then move on to a more complex situation of a two-qudit Hilbert space $\mathcal{H}_{2,d}=\mathbbm{C}^d\otimes\mathbbm{C}^d$ 
and consider subspaces  $V\subset \mathcal{H}_{2,d}$ spanned by $k$ entangled orthogonal pure states $\ket{\phi_i}$.

We start off by considering the case of the  geometric measure of entanglement.
The maximal dimension of
a CES that can be detected with Fact \ref{criterion} is bounded as 
$\dim V<d$, while the maximal CES is of dimension $(d-1)^2$. To see this explicitly let us 
assume that all vectors $\ket{\phi_i}$ have the same GM, denoted $E'$. Then, the criterion detects complete entanglement of that subspace if $k< 1/(1-E')$. Taking $E'$ to be the maximal GM achievable in $\mathcal{H}_{2,d}$, i.e., $E'=(d-1)/d$, the condition gives $k<d$. 
This has an immediate implication that any subspace of $\mathcal{H}_{2,d}$ spanned by $k<d$ mutually orthogonal maximally entangled  vectors is a 
CES. To give an example of such a subspace consider the following set of $d$ vectors
\begin{equation}\label{vectors}
\ket{\psi_j}=\frac{1}{\sqrt{d}}\sum_{i=0}^{d-1}\omega^{ij}\ket{ii}, \quad \omega=\mathrm{exp}(2\pi\mathbbm{i}/d), 
\end{equation}
for $j=0,\ldots,d-1$.
By the argument above, any $k$-element subset of them with $k<d$ will span a CES whose entanglement is lower-bounded as $E_{\min}(V)\geq 1-k/d$. 
Moreover,  the $d$-dimensional subspace spanned by all of these vectors clearly contains a product vector, which can be obtained by simply constructing an equal superposition of all $\ket{\psi_j}$'s, that is,
$(1/\sqrt{d})\sum_{j=0}^{d-1}\ket{\psi_j}=\ket{00}$. This implies that the criterion is in this sense tight.

Further, we observe  that
  criterion (\ref{kryterium}) can be reformulated in terms of 
the Schmidt coefficients. Precisely, with the aid of  formula (\ref{KyFan}) we obtain the following statement: {\it let $\lambda_{j}^i$ be the Schmidt coefficients of $k$ basis states  $\ket{\varphi_i}$ of $V$; further, let $r_i$ be their Schmidt ranks; if}
\begin{equation}
 \sum_{i=1}^k \sum_{j=1}^{r_i}  (\lambda_{j}^i)^2 < 1 
\end{equation}
{\it then $V$ is a CES}.

As an example to the above statement, let us consider again vectors (\ref{vectors}). For them $E_r=1-(r-1)/d$, and hence for any subset of $k$ vectors Eq. (\ref{bound-min-subspace}) gives
\begin{equation}
\mathcal{E}_{\min}^r(V)\geq 1-k\frac{r-1}{d}.
\end{equation}
The right-hand side exceeds zero if $k<d/(r-1)$, and consequently, 
any $k$-element subset of $d$ vectors (\ref{vectors}) with $k<d/(r-1)$ spans
a CES in which all vectors have Schmidt rank at least $r$. In particular,
in the extremal case $r=d$, the criterion detects only one-dimensional 
subspaces.  Clearly, in the case $r=2$ we recover what we have previously established for these vectors.

Let us now see how our criterion relates to the one obtained by Gour and Roy \cite{Gour}. Their result is the following: let $V$ be spanned by $k$ bipartite orthogonal $\psi_i$'s with the Schmidt ranks $r(\psi_1)\le r(\psi_2) \le \dots \le r(\psi_k)$; then 
\begin{equation}\label{Gour}
r_{\min}(V)\geq \min_{m=1,2,\dots,k}\{r(\psi_m) - \sum_{i=1}^{m-1} r(\psi_i)   \},
\end{equation}
where $r_{\min}(V)$ is the smallest Schmidt rank among the states from $V$.

First, an immediate observation is that the condition (\ref{Gour}) cannot detect subspaces for any local dimensions and its utility is limited to rather large subspaces with $d\ge 2^k$. On the other hand, our result is applicable to any $d\ge 3$ making it in this sense more universal. Second, certain combinations of the Schmidt ranks of the basis states lead to a trivial bound (zero or one), while our bound is free from this disadvantage. For example, condition (\ref{Gour}) cannot detect entanglement of two dimensional subspaces spanned by states with equal Schmidt ranks, whereas relying on (\ref{kryterium}) we have shown above that any two states of the form (\ref{vectors}) do span a CES. To make the comparison fair, however, we note that this does not mean that our criterion is stronger in general - there exist subspaces not detected by our criterion but detected by (\ref{Gour}). Such an exemplary subspace in $\mathbbm{C}^4\otimes \mathbbm{C}^4$ is spanned by the vectors: $\ket{\psi_1}=1/\sqrt{2}(\ket{22}-\ket{33})$ and $\ket{\psi_2}=2/\sqrt{7}\ket{00}+1/\sqrt{7}(\ket{11}+\ket{22}+\ket{33})$.

\section{Applications in multipartite case}

Let us now move to the richer multipartite case and consider both CESs and GESs. We put particular emphasis on subspaces of the symmetric subspaces as those are most important from the practical point of view. 
Except a simple case in the next paragraph, in the following examples we check the conditions of the criterion for one -- the most natural -- basis.

\subsection{Completely entangled subspaces}\label{sec:CES}

It is important that in the multipartite case the restriction $k<d$ no longer holds if we are interested in the sole fact whether the subspaces are entangled or not, meaning that qubit subspaces are also detectable. This can be seen for example by considering the subspace spanned by the GHZ state and the W state, for which the GM is, respectively, $1/2$ and $5/9$.  
The latter subspace can also be used to illustrate benefits of using different basis to check condition (\ref{kryterium}) or (\ref{kryterium-multiparty}) . If we take its spanning vectors as $\ket{\phi_1}=\sqrt{2/5}|GHZ\rangle+\sqrt{3/5} |W\rangle$ and $\ket{\phi_2}=\sqrt{3/5}|GHZ\rangle-\sqrt{2/5} |W\rangle$, then the bound from 
Fact \ref{fakt-glowny} is trivial. 

\subsubsection{$N$--qubit GHZ and W states}

Let us now consider a more general case already studied in the literature, namely that of 
two-dimensional subspaces $V_{N,2}$
spanned by the $N$-qubit GHZ state and the $W$ state, which is simply $\ket{D_{N,1}}$. It is known that $E(\ket{\mathrm{GHZ}_{N,2}})=1/2$, whereas $E(\ket{D_{N,1}})=1-[(N-1)/N]^{N-1}$ \cite{GME2}. We then have 
\begin{equation}
E_{\min}(V_{N,2})\geq \frac{1}{2}-\left(\frac{N-1}{N}\right)^{N-1}.
\end{equation}
It is not difficult to see that the right-hand side of the above
is greater than zero for any $N\geq 3$ 
and thus these two states span a CES for any $N\geq 3$.

\subsubsection{Qubit Dicke states}

We now investigate subspaces of the symmetric subspace, that is subspaces spanned by the Dicke states $\ket{D_{N,k}}$. In general it holds \cite{Wei_Dicke}

\begin{equation}\label{Raimat}
E_{GM}(\ket{D_{N,k}})=1-\binom{N}{k}\left(\frac{k}{N}\right)^k\left(\frac{N-k}{N}\right)^{N-k}.
\end{equation}

For any number of qubits $N$, states with $k = 0, N$ are fully separable.
On the other hand, it easy to verify that all other states are entangled and the subspace spanned by $\ket{D_{N,k}}$, $k=1,2,\dots, N-1$, is a CES and its entanglement equals $1/2^{N-1}$. 
In what follows we will try to establish how large a subspace of this subspace can be detected with our criterion.

Clearly, the states with the highest entanglement are those with $k=N/2$ (even $N$) or $k=\frac{N\pm1}{2}$ (odd $N$) and the entanglement of the remaining states is a decreasing function when $k$ moves away from these values.
One thus needs to analyze how many states around the said central $k$ can be used to create a CES detectable by the criterion. 
To illustrate the approach, we assume that $N$ is even and
 consider a set $\Pi_{m} = \{\ket{D_{N,m}},\dots,\ket{D_{N,N/2}} ,\dots,\ket{D_{N,N-m}}\}$  of $|\Pi_m|=N-2m+1$ Dicke states. The expressions encountered in the calculations are intractable analytically in full generality and we will use approximations
 to establish the value of $m$ above which subspaces spanned by $\Pi_m$ are certainly CESs by Fact \ref{criterion}. This will lead to a lower bound on the dimension of the largest CES that is in fact detected by the criterion for a given $N$.
 
 Condition (\ref{kryterium}) applied to $\Pi_m$ yields
 \begin{eqnarray}\label{trudny-warunek}
 \sum_{k=m}^{N-m}E_{GM}(\ket{D_{N,k}}) -(|\Pi_m|-1)>0
 \end{eqnarray}
with $E_{GM}$ of the Dicke states given by (\ref{Raimat}).
 It is verv difficult to find a closed--form expression for the sum for any number of states  (it is plausible that no such form exists) and from this infer the actual threshold value of $m$.
Observe, however, that the sum over $\Pi_m$ can  be bounded as follows 
\begin{equation}
\begin{split}
& \sum_{k=m}^{N-m} E_{GM}(\ket{D_{N,k}}) > |\Pi_m| \times E_{GM}(\ket{D_{N,m}}) = \\  
& |\Pi_m|  \left[1 - \binom{N}{m} \bigg(\frac{m}{N}\bigg)^m \bigg(\frac{N-m}{N}\bigg)^{N-m} \right] \\
&> |\Pi_m| \left[1-  \frac{e}{2\pi} \sqrt{\frac{N}{m(N-m)}} \right]   . 
\end{split}
\end{equation}
where in the first inequality we have bounded the entanglement of each state by the entanglement of the least entangled one, i.e., $\ket{D_{N,m}}$, while in the second inequality we have used the bounds linked to the Stirling approximation: $\sqrt{2\pi} n^{n+\frac{1}{2}}e^{-n}<n!<e n^{n+\frac{1}{2}}e^{-n}$.
If we now use this value in (\ref{trudny-warunek}) instead of the exact one, we can still satisfy the inequality for some  $m$. In this manner we will obtain a condition yielding a value of $m$ larger than the true threshold but also guaranteeing that the subspace is a CES (yet smaller than the optimal one).  There follows 
that if the following inequality is satisfied, then the subspace spanned by $\Pi_m$ is a CES
\begin{equation}\label{eq:approx_Dicke}
   (N-2m+1) \frac{e}{2\pi} \sqrt{\frac{N}{m(N-m)}}  < 1,
\end{equation}
where we have plugged back in $|\Pi_m|=N-2m+1$.
Solving this for $m$ we obtain:
\begin{equation}\label{eq:bound_Dicke}
    m > 
    \frac{N}{2}-\frac{\pi  \sqrt{N \left(e^2 \left(N^2-1\right)+\pi ^2 N\right)}-e^2 N}{2 \left(e^2 N+\pi ^2\right)},
\end{equation}
which for large $N$ is approximated by 
\begin{equation}
    m > \frac{N}{2} - \frac{\pi}{2e}\sqrt{N},
\end{equation}
so the deviation from the central $k = N/2$ is of the square root order. In consequence, for large $N$, subspaces with dimensions of order $\sqrt{N}$ are detected.

In view of the approximations used, one might now wonder how good the approximation given in (\ref{eq:bound_Dicke}) is in comparison to the true value stemming from (\ref{trudny-warunek}), which can be obtained by direct numerical summation and verification whether the inequality holds for a given $m$.
Accuracy of the analytical bound (\ref{eq:bound_Dicke}) is presented in Fig.~\ref{fig:Dicke_bound}.

\begin{figure}[h!]
	\includegraphics[width=8.5cm]{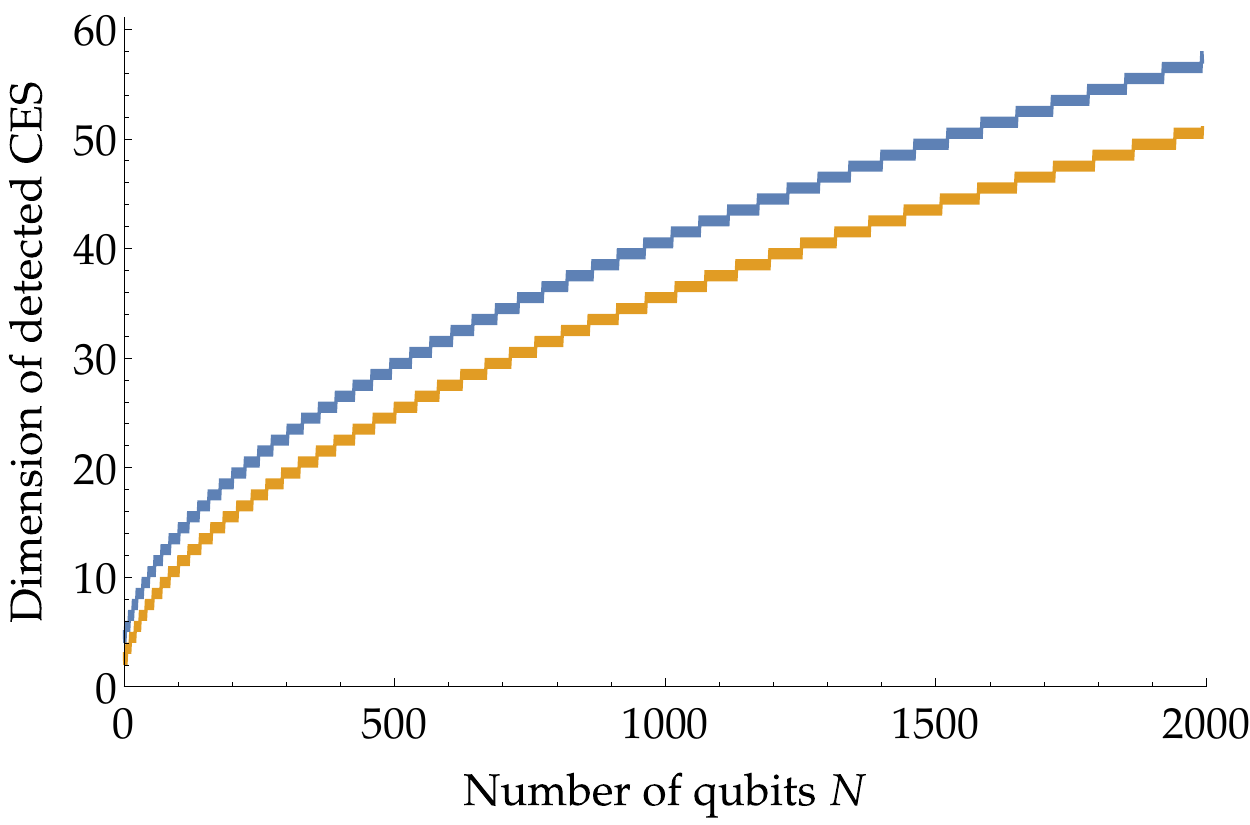}
	\caption{Power of the criterion for identifying subspaces spanned by qubit Dicke states as being CESs. Upper curve gives the dimensions of the largest detectable CESs and corresponds to straightforward summation over the vectors closest to the central one \mbox{($k = N/2$)} in (\ref{trudny-warunek}), while the bottom one denotes an analytical bound on the dimensions following from Eq.~(\ref{eq:approx_Dicke}).}\label{fig:Dicke_bound}
\end{figure}

\subsubsection{Antisymmetric subspace}
As a complementary example, let us now consider the antisymmetric subspace of $\mathcal{H}_{N,d}$, which is known to be a CES (in fact a GES too, see the upcoming section), and check the power of the criterion in identifying it as such. It is known that  the entanglement of any vector is given by $E_{\mathrm{GM}}=1-1/N!$. Thus, the criterion gives the condition on the number of parties and local dimensions ${d \choose N} < N!$. One can check analytically that it certainly detects whenever $N \ge \left \lfloor{\frac{d+1}{2}}\right \rfloor +1$, although this is far from optimal. The results are plotted in  Fig.~\ref{antysym-entanglement} with an additional bound, which works remarkably good in the considered region (there are just a few single points below $d=50$ falling under this bound in the region of interest). We can see that the criterion is quite powerful for the CES case of the antisymmetric subspace.

\begin{figure}[h!]
	\includegraphics[width=8.5cm]{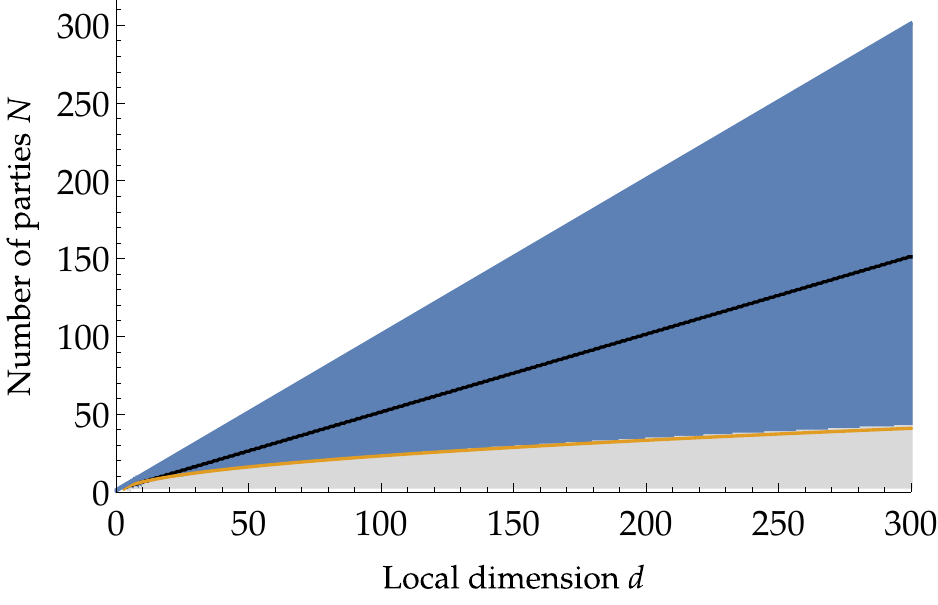}
	\caption{Power of the criterion for identifying the antisymmetric subspace as being a CES. Blue region - detected by the criterion, grey -- not detected, black line denotes $N=\left \lfloor{\frac{d+1}{2}}\right \rfloor +1$, orange curve corresponds to $\sqrt{2d+2}-\sqrt{d-3}-1$.
	}\label{antysym-entanglement}
\end{figure}

\subsubsection{Error correction codes}
The criterion could also be used to verify that certain error correction codes give rise to  entangled subspaces. For example, the geometric measure of entanglement of the two codewords of the seven--qubit Steane code is known to be  $7/8$ \cite{Markham}, immediately implying that the subspace is a CES.

\subsection{Genuinely entangled subspaces}

Let us now investigate the power of our criterion in 
detecting genuinely entangled subspaces using the generalized geometric measure of entanglement.

First, we note that its usefulness is limited to  $d\geq 3$ as for any $N$-qubit pure state $\ket{\psi}$ it holds $E_{GGM}(\ket{\psi})\leq 1/2$.
Second, an observation as in the bipartite case follows that it must be 
$k< d$.
%
To obtain this bound we need to consider all bipartitions of the parties and find the  value of the GM for them, then pick the cut with the smallest value. 
This ''least'' entangled cut is  $1|N-1$ parties, which stems from the fact that bipartite states from  $\mathbbm{C}^d\otimes (\mathbbm{C}^{d})^{\otimes N-1} $ are in fact supported on $\mathbbm{C}^d\otimes \mathbbm{C}^d $ (this again follows from the Schmidt decomposition), and them the maximal amount of entanglement as measured by the GM is $1-1/d$. The result then follows.

\subsubsection{Generalized GHZ}
To provide the first example let us consider  the generalized GHZ states of the form
\begin{equation}\label{gen-GHZ}
\ket{\mathrm{GHZ}'_j}=\frac{1}{\sqrt{d}}\sum_{i=0}^{d-1}\ket{i}^{\otimes N-1}\otimes\ket{i+j},
\end{equation}
where $j=0,\ldots,d-1$ and addition modulo $d$ in the last ket. 
Now, $E_{GGM}$ of any such a state is $(d-1)/d$, and therefore a subspace
spanned by any $k<d$ vectors (\ref{gen-GHZ}) is genuinely entangled (by Fact \ref{ent-criterion} this also applies to mixtures of such states, cf. \cite{FloresGalapon}). Clearly, an equal superposition of all such states gives a biproduct vector.

\subsubsection{Absolutely maximally entangled states}
Absolutely maximally entangled (AME) states are widely studied multipartite states exhibiting in some sense the strongest entanglement in a given system of $N$ qudits  \cite{AME}.
A state $\ket{\psi}\in \mathcal{H}_{N,d}$ is called AME, denoted AME$(N,d)$, iff all reductions of at least half of the subsystems yield the maximally mixed state, i.e., the partial trace $\tr_S \ket{\psi}\bra{\psi} \propto \mathbb{I}$, for any subsystem $S$ of $|S|=\lceil N/2 \rceil$ parties.

Existence of an AME state depends on the local dimension $d$ as well as on the number of parties $N$ (see \cite{table_AME} for the current state of art).
Importantly, if an AME state exists for a given pair $(N,d)$, there follows existence of an orthonormal basis of the whole Hilbert space $\mathcal{H}_{N,d}$ composed of $d^N$ orthogonal AME($N,d$) states \cite{Raissi_2020}.
This observation allows us to use Fact \ref{criterion} to determine how big a GES can be created using AME states. With this aim we only need to know the GGM of AME states, which is easily found to be $E_{GGM}(\ket{\text{AME}_{N,d}}) = 1-1/d$. This is because AME states are maximally entangled across any cut and as such achieve the maximal value of the GM for all cuts. Taking the minimum over all bipartitions we arrive at the claimed value.
Therefore, if in a given quantum system described by $\mathcal{H}^{\otimes N}_{d}$ there exists an AME state, there are also genuinely entangled subspaces of dimensions up to $d-1$, spanned by any set of orthogonal AME states.

\subsubsection{Qudit Dicke states}
Let us now explore the case of genuine entanglement of subspaces of the symmetric subspaces spanned by the qudit Dicke states

\begin{equation}\label{general-Dicke}
\ket{D_{N,\vec{k}}^d}=\sqrt{\frac{\Pi_{i=0}^{d-1}k_i!}{N!}}\sum_p \sigma_p\left(\ket{0}^{\otimes k_0} \ket{1}^{\otimes k_1} \cdots \ket{d-1}^{\otimes k_{d-1}}\right)
\end{equation}
with the sum over distinct permutations.

In Appendix \ref{d-Dicke} we show that the GGM of the Dicke states equals:
	\begin{eqnarray}\label{GM-dicke}
E_{GGM}\left( \ket{D_{N,\vec{k}}^d}     \right) &=&\nonumber \\
&&\hspace{-3 cm} =  \min_{n,\vec{\pi}} \left[ 1- {N \choose n}^{-1} {k_0 \choose \pi_0}  {k_1 \choose \pi_1} \cdots  {k_{d-1} \choose \pi_{d-1}} \right],
\end{eqnarray}
where it holds $\pi_i \le k_i$ and $\sum_i \pi_i=n$.	

There are no known closed forms for the minimisation over $\vec{\pi}$ and thus we have verified by direct search the largest dimensions of detectable GESs composed of $d$--level Dicke states. The results are presented in Fig. \ref{GES-Dicke-plot}.

Nevertheless, it is easy to obtain a bound on the largest detectable GES. The most entangled states for given $d$ have the GGM equal to $1-1/N$. They correspond to vectors $\vec{k}$ with all $k_i$ equal and they exist whenever $d \ge N$. If $d>N$ there are ${d \choose N}$ such vectors and at most $N-1$ of them span a GES according to the criterion. 

\begin{figure}
	\includegraphics[width=8.5cm]{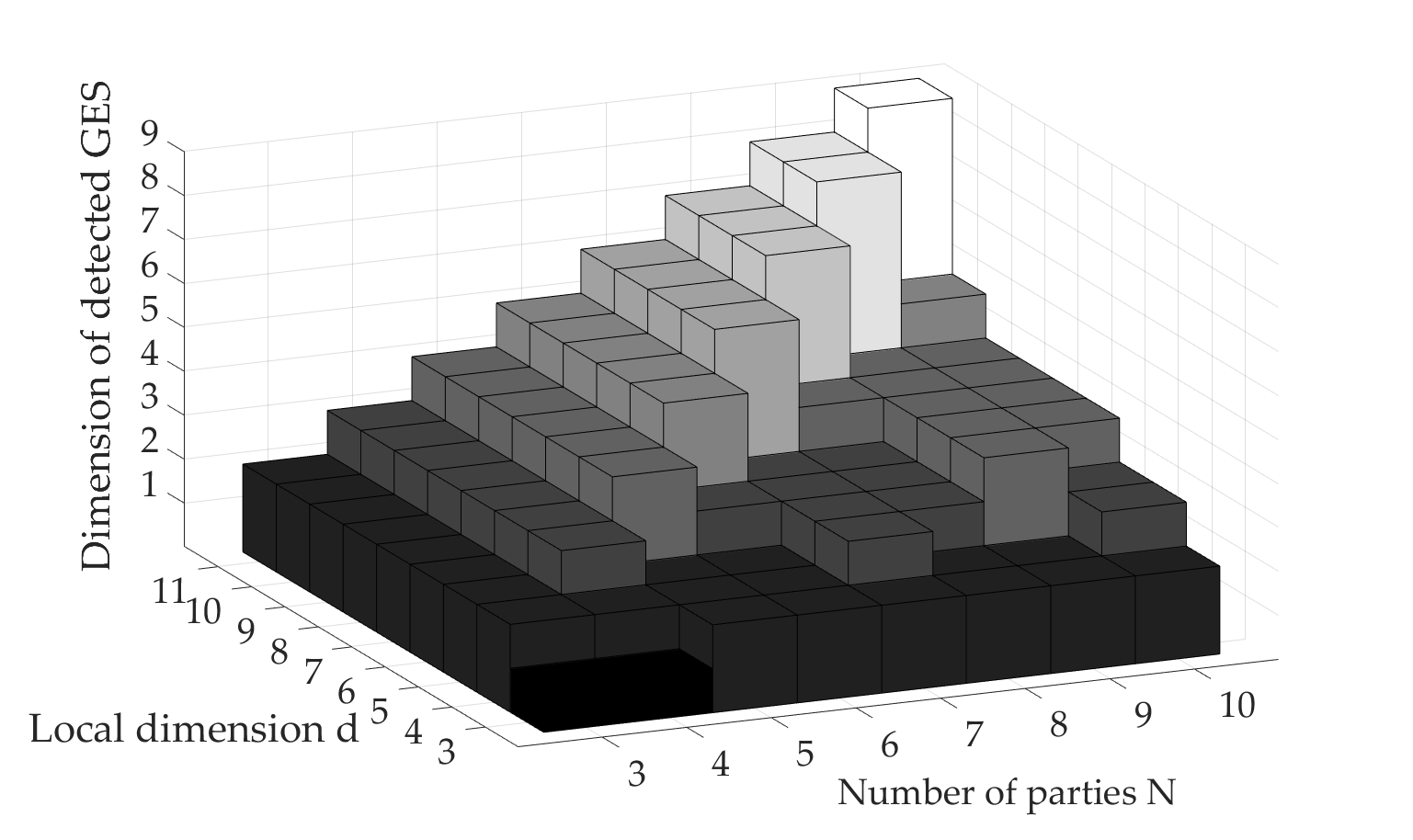}
	\caption{Dimensions of subspaces spanned by the qudit Dicke states identified  as GESs by the criterion for $3\le N \le 10$ and $3 \le d \le 11$.}\label{GES-Dicke-plot}
\end{figure}

\subsubsection{Antisymmetric subspace}

The case of the antisymetric subspace can be immediately solved.
It is known that for any basis vector it holds  $E_{\mathrm{GGM}}=1-1/N$. This implies that any subspace of dimension $N-1$ or less of the antisymmetric subspace is identified as a GES by the criterion.

Concluding this section, we note that the fact that the criterion is weaker in the GES case can be attributed to the fact that in general the GGM is smaller than the GM, while the RHS of (\ref{kryterium}) is the same in both cases.

\section{Conclusions and outlook}

We have considered the problem of judging whether a subspace is entangled or not, both in the completely and genuinely entangled case, on the basis of the amount of entanglement of the basis states. With this aim we have provided a simple sufficient criterion and illustrated the approach with several examples, in particular the all-important symmetric and antisymmetric subspaces. We have also shown that the  condition directly leads to an entanglement criterion for mixed states.

Future work could concern extending the results given here to other entanglement measures. Also, the problem of bounding the maximal subspace entanglement in terms of geometric or other measures deserves separate treatment (cf. \cite{Gour}).
It might also be interesting to look into the possibility of developing other entanglement criteria for mixed states based on the entanglement of the states from the mixture.

\section{Acknowledgments}

GRM is supported by National Science Centre (Poland) under the Maestro grant number DEC2015/18/A/ST2/00274. RA acknowledges the support from the Polish National Science Center through the SONATA BIS project No. 2019/34/E/ST2/00369.

\appendix

\section{Proof of Eq. (\ref{KyFan})} \label{schmidt-proof}

Let the given state that we want to find $E_r$ of, be
\begin{equation}\label{schmidt-app}
\ket{\psi}=\sum_{i=1}^{R}\lambda_i\ket{e_i}\ket{f_i}
\end{equation}
with the Schdmidt coefficients $\lambda_i$ ordered decreasingly and orthonormal bases $\{\ket{e_i}\}_{i=1}^d$ and $\{\ket{f_i}\}_{i=1}^d$ [cf. (\ref{schmidt})]. Assume  $\lambda_{R+1}=\lambda_{R+2}=\dots=\lambda_d=0$.

We optimze (\ref{r-schmidt}) over states of the form
\begin{equation}\label{optymalny}
\ket{\varphi}=\sum_{i,j=1}^d a_{ij}\ket{e_i}\ket{f_j}
\end{equation}
with rank $r-1$ matrix $[A]_{ij}=a_{ij}$. Let its non-zero singular values be $s_1 \ge s_2 \ge \dots \ge s_{r-1}$.

 Due to the von Neumann trace inequality, which states that for complex $n\times n$ matrices $A$ and $B$ with singular values, respectively, $a_1 \ge a_2 \ge \cdots a_n$ and $b_1 \ge b_2 \ge \cdots b_n$, it holds $|\tr AB| \le \sum_i  a_i b_i$,  we have 
\begin{equation}\label{}
|\bra{\psi}\varphi \rangle | =  \left|\sum_{i=1}^d a_{ii}\lambda_i\right| \le \sum_{i=1}^{r-1} s_i\lambda_i,
\end{equation}
which, under the constraint $\sum_i s_i^2 =1$ (normalization of $\ket{\varphi}$), is clearly maximized with the choice  
\begin{equation}\label{}
s_i=\frac{\lambda_i}{\sqrt{\sum_{j=1}^{r-1}\lambda_j^2}}.
\end{equation}
We deduce that the optimal state (\ref{optymalny}) is simply 
\begin{equation}\label{optimal}
\ket{\varphi}=\left(\sum_{j=1}^{r-1}\lambda_j^2\right)^{-\frac{1}{2}}\times\sum_{i=1}^{r-1} \lambda_{i} \ket{e_i}\ket{f_i}.
\end{equation}
The results then follows.

\section{Proof of Fact 1}

\begin{proof}
Consider a superposition of $k$ pure mutually orthogonal states
$\ket{\phi_i}$
\begin{equation}
    \ket{\Psi}=\sum_{i=1}^k\alpha_i\ket{\phi_i},
\end{equation}
and recall that all the considered entanglement quantifiers can be wrapped up 
in a single formula
\begin{equation}\label{AppA:f1}
    \mathcal{E}(\ket{\psi})=1-\max_{\ket{\varphi}\in \mathcal{S}}|\langle\varphi|\psi\rangle|^2,
\end{equation}
where $\mathcal{S}$ is any set considered in the main text. 

Due to the triangle inequality $|x+y|\leq |x|+|y|$, the expression under the maximum on the right-hand side of the above 
for the superposition $\ket{\Psi}$ can be upper bounded as
\begin{eqnarray}
    |\langle\varphi|\Psi\rangle|^2&\leq& \left(\sum_{i=1}^k|\alpha_i||\langle \varphi|\phi_i\rangle|\right)^2\nonumber\\
    &=& \sum_{i=1}^k|\alpha_i|^2 |\langle \varphi|\phi_i\rangle|^2+2\sum_{i<j}|\alpha_i\alpha_j||\langle \varphi|\phi_i\rangle \langle \varphi|\phi_j\rangle|,\nonumber\\
\end{eqnarray}
which holds for any $\varphi$.
Plugging this into Eq. (\ref{AppA:f1}) and using the fact that 
$\sum_i|\alpha_i|^2=1$, we obtain
\begin{eqnarray}\label{AppA:f2}
    \mathcal{E}(\ket{\Psi})&\geq& \sum_{i=1}^k |\alpha_i|^2\mathcal{E}(\ket{\phi_i})-2\max_{\ket{\varphi}\in \mathcal{S}}\sum_{i<j}|\alpha_i\alpha_j||\langle \varphi|\phi_i\rangle \langle \varphi|\phi_j\rangle| \nonumber \\
    &\ge & \sum_{i=1}^k |\alpha_i|^2\mathcal{E}(\ket{\phi_i}) \nonumber \\
    && \hspace{+0.5cm} -2 \sum_{i<j} |\alpha_i\alpha_j|\max_{\ket{\varphi}\in \mathcal{S}}|\langle \varphi|\phi_i\rangle|\max_{\ket{\varphi}\in \mathcal{S}}| \langle \varphi|\phi_j\rangle|,
\end{eqnarray}
where in the second inequality we have first exploited the fact that the maximum of the sum is upper bounded by the sum of maxima, and then bounded from above each maximum of  products by the product of maxima.
With the aid of the fact that  $ \max_{\ket{\varphi}\in\mathcal{S}}|\langle \varphi|\phi_i\rangle| = \sqrt{1-\mathcal{E}(\ket{\phi_i})}$ this gives the claimed inequality.
%
%
\end{proof}
\label{AppA}

\section{ Proof of Fact 2}
\label{AppB}

\begin{proof}
With the aid of inequality~(\ref{bound-superpozycja}) we can bound $\mathcal{E}_{\min}(V)$ from below as
\begin{equation}\label{bound}
\mathcal{E}_{\min}(V)\equiv \min_{\ket{\psi}\in V}\mathcal{E}(\ket{\psi})\geq \min_{\substack{a_i\geq 0\\a_1^2+\ldots+a_k^2=1}}\widetilde{E}(a_1,\ldots,a_k),
\end{equation}
where
\begin{equation}
\widetilde{E}(a_1,\ldots,a_k)=\sum_{i=1}^{k}a_i^2 \mathcal{E}_i-2\sum_{i<j}a_ia_j
\sqrt{1-\mathcal{E}_i}\sqrt{1-\mathcal{E}_j}
\end{equation}
with $\mathcal{E}_i :=\mathcal{E}(\ket{\phi_i})$.

The key observation now is that $\widetilde{E}$ can be 
conveniently rewritten in the following simple form
\begin{equation}\label{function}
\widetilde{E}(a_1,\ldots,a_k)=1-N\langle v|a\rangle^2,
\end{equation}
where $\ket{a}$ is a normalized vector of variables, 
$\ket{a}=(a_1,\ldots,a_k)^T\in\mathbbm{R}^k$, and  
\begin{equation}
\ket{v}=(1/\sqrt{N})(\sqrt{1-\mathcal{E}_1},\ldots,\sqrt{1-\mathcal{E}_k})^T
\end{equation}
with $N$ being the normalization constant defined as
$N=k-\sum_{i}\mathcal{E}_i$. It is clear that the minimal value of the function defined in Eq.~(\ref{function}) is attained for $\ket{a}=\ket{v}$, that is when
\begin{equation}
a_i=\frac{\sqrt{1-\mathcal{E}_i}}{\sqrt{k-\sum_{i}\mathcal{E}_i}},
\end{equation}
and it is given by 
\begin{equation}
\widetilde{E}(a_1,\ldots,a_k)=\sum_{i=1}^k\mathcal{E}_i-(k-1),
\end{equation}
which completes the proof.
\end{proof}

The above proof exploits directly a bound on a superposition of states. There is also another proof, more straightforward, which avoids this and
uses the following result from \cite{Zhu}:
\beq\label{branciardowy}
\mathcal{E}_{\mathrm{min}}(V)=1- \max_{\ket{\varphi}\in \mathcal{S}} \bra{\varphi} \mathcal{V} \ket{\varphi},
\eeq
where $\mathcal{V}$ is the projection onto subspace $V$.
Originally, this formula was derived for the geometric measure of entanglement, but it is easy enough to realize that it also applies to any of the measures considered in the present paper, defined through the properties of the set $\mathcal{S}$.

Let $\mathcal{V}=\sum_{i=1}^k{\proj{\phi_i}}$. Inserting this to (\ref{branciardowy}) we obtain
\begin{eqnarray}
\mathcal{E}_{\mathrm{min}}(V)&=& 1- \max_{\ket{\varphi}\in \mathcal{S}} \bra{\varphi} \sum_{i=1}^k\proj{\phi_i} \ket{\varphi}\nonumber\\
&\ge&  1- \sum_{i=1}^k\max_{\ket{\varphi}\in \mathcal{S}} \bra{\varphi} \phi_i\rangle \langle\phi_i \ket{\varphi}\nonumber\\
&=& 1- \sum_{i=1}^k (1-\mathcal{E}(\phi_i)) \nonumber \\
&=& \sum_{i=1}^k\mathcal{E}(\phi_i) -(k-1).
\end{eqnarray}

\section{GGM of $d$--level Dicke states}\label{d-Dicke}

We start by proving the following statement about the entanglement of $d$--level Dicke states. The result applies to any $d$, including the qubit case $d=2$ for which it agrees with the known formula for the GGM \cite{GGM-toth,GGM-bergmann,GGM-hindusi}.

\begin{fakt}
	The GM of $\ket{D_{N,\vec{k}}^d} $ across a cut $n|N-n$ ($n=1,2,\dots, \lfloor N/2 \rfloor$) is given by
	\begin{widetext}
		\begin{eqnarray}\label{general-GM}
		E_{GM}^{n|N-n}\left( \ket{D_{N,\vec{k}}^d}     \right) &=& 1- \frac{k_0!k_1!\cdots k_{d-1}!}{N!} \max_{\vec{\pi}} \frac{n!}{\pi_0! \pi_1!\cdots \pi_{d-1}!} \frac{(N-n)!}{(k_0-\pi_0)! (k_1-\pi_1)!\cdots (k_{d-1}-\pi_{d-1})!}\\
		&=& 1- {N \choose \vec{k}}^{-1} \max_{\vec{\pi}} {n \choose \vec{\pi}}{N-n \choose \vec{k}-\vec{\pi}} \label{multinomial}\\
		&=& 1- {N \choose n}^{-1} \max_{\vec{\pi}} {k_0 \choose \pi_0}  {k_1 \choose \pi_1} \cdots  {k_{d-1} \choose \pi_{d-1}} \label{binomial} \\
			&=&  \min_{\vec{\pi}} \left[ 1- {N \choose n}^{-1} {k_0 \choose \pi_0}  {k_1 \choose \pi_1} \cdots  {k_{d-1} \choose \pi_{d-1}} \right].\label{binomial-global-min} 
		\end{eqnarray}
	\end{widetext}
	where ${a \choose \vec{b}}= \frac{a!}{b_0!b_1! \cdots b_{d-1}!}$ is the multinomial coefficient, for any $i$ it holds $\pi_i \le k_i$ and $\sum_i \pi_i=n$.	
	
	The GGM equals
	\begin{equation}\label{general-GGM}
	E_{GGM}\left( \ket{D_{N,\vec{k}}^d}     \right) =\min_n E_{GM}^{n|N-n}\left( \ket{D_{N,\vec{k}}^d}     \right). 
	\end{equation}
\end{fakt}
\begin{proof}
	Given $\ket{D_{N,\vec{k}}^d} $ we write it in its Schmidt decomposition for the cut $n|N-n$ which is simply
	\begin{equation}\label{}
	\ket{D_{N,\vec{k}}^d} = {N \choose \vec{k}}^{-\frac{1}{2}} \sum_{\vec{\pi}} \ket{\tilde{D}_{n,\vec{\pi}}^d}\ket{\tilde{D}_{N-n,\vec{k}-\vec{\pi}}^d},
	\end{equation}
	where $\ket{\tilde{D}_{a,\vec{b}}^d}$ is the unnormalized state $\ket{D_{a,\vec{b}}^d}$, i.e., $\ket{\tilde{D}_{a,\vec{b}}^d}= \sqrt{{a \choose \vec{b}}} \ket{D_{a,\vec{b}}^c}$ . It is evident that the optimal state for the computation of the GM is a product of Dicke states $\ket{D_{n,\vec{\pi}}^d}\ket{D_{N-n,\vec{k}-\vec{\pi}}^d}$ with the largest number of terms, which for a given cut is determined by $\vec{\pi}$. Equation (\ref{multinomial}) follows.  The transistion from (\ref{multinomial}) to  (\ref{binomial}) is simple algebra.

	The second statement is obvious and it follows from the very definition of the GGM.
\end{proof}

Let us note that we encounter here multivariate hypergeometric distribution.  The interpretation of the term with the minus sign in (\ref{multinomial}-\ref{binomial-global-min}) is the following: we have $N$ objects of $d$ types, among which there is $k_i$ objects of the $i$--th type; we randomly draw $n$ objects; the said term gives us the probability of drawing $\pi_i$ objects of type $i$.
Although appealing, this identification does not prove very useful as, to our knowledge, there are no closed forms for the maximisations. For this reason, in the computation of the detectable GESs in the general case we have used direct search through all the possible cases of $\vec{k}$, $\vec{\pi}$ and $n$ for a given $(N,d)$. The results are presented in the main text in Fig. \ref{GES-Dicke-plot}.

Here, for illustration purposes, we plot the GGM of the most entangled states as a function of $N$  for chosen values of $d=3,4,5$. 
\begin{figure}[H]\label{max-GGM-Dicke}
	\includegraphics[width=8.5cm]{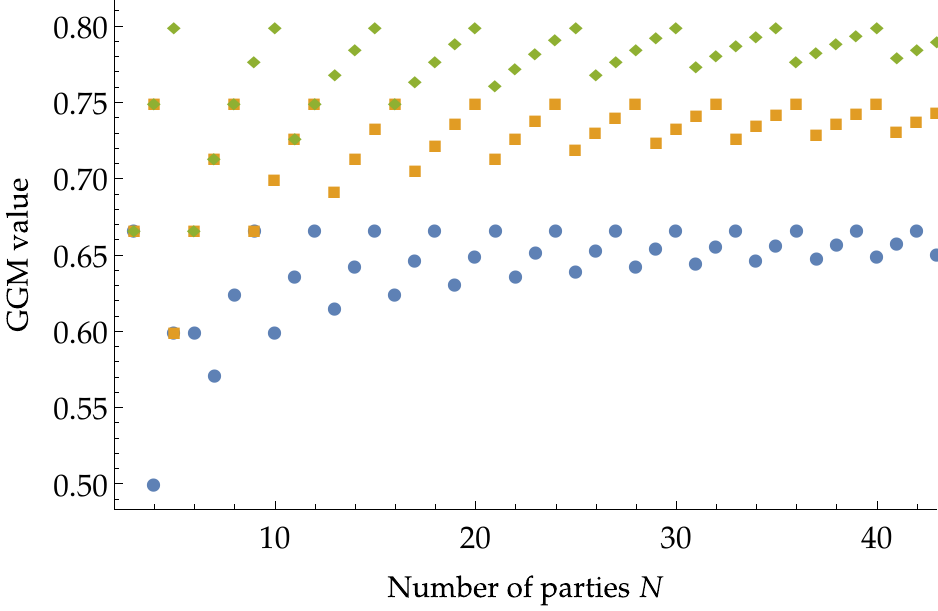}
	\caption{The GGM of the most entangled $d$--level Dicke states as a function of $N$. Blue points correspond to $d=3$, orange --- $d=4$, green --- $d=5$.}
\end{figure}

\end{document}